\documentclass[envcountsame]{llncs}
\usepackage{amsmath,amssymb,euscript}
\usepackage{relsize} 
\usepackage{tikz}
\usetikzlibrary{shapes}

\usepackage{hyperref}

\def\ca#1{{\mathcal#1}}

\def\mx#1{{\vec#1}}
\def\cO{{\ca O}}
\def\pbl#1{{\sc#1}}
\newcommand{\crg}{\mathrm{cr}}
\newcommand{\clg}{\mathrm{cl}}
\newcommand{\crgnc}{\mathrm{cr^n}}

\begin{document}
	
\title{Exact Crossing Number Parameterized \mbox{by Vertex Cover}}

\author{Petr Hlin\v en\'y\,\inst{1}\orcidID{0000-0003-2125-1514}%
\thanks{Supported by the {Czech Science Foundation}, project no.~{17-00837S}.}
\thanks{Corresponding author: Petr Hlin\v en\'y} \and
\mbox{Abhisekh Sankaran}\,\inst{2}
\thanks{Supported by the Leverhulme Trust through a Research Project Grant
	on `Logical Fractals'.}
}

\authorrunning{P.~Hlin\v{e}n\'y and A.~Sankaran}

\institute{Faculty of Informatics of Masaryk University, Brno, Czech Republic
\\\email{hlineny@fi.muni.cz} 
\smallskip \and
Department of Computer Science and Technology, University of Cambridge, UK
\\\email{abhisekh.sankaran@cl.cam.ac.uk}}

	\maketitle

\begin{abstract}
We prove that the exact crossing number of a graph can be efficiently
computed for simple graphs having bounded vertex cover.
In more precise words, \pbl{Crossing Number} is in FPT when parameterized by the
vertex cover size.
This is a notable advance since we know only {\em very few} nontrivial 
examples of graph classes with unbounded and yet efficiently computable
crossing number.
Our result can be viewed as a
strengthening of a previous result of Lokshtanov [arXiv, 2015] that
\pbl{Optimal Linear Arrangement} is in FPT when parameterized by the
vertex cover size, and we use a similar approach of reducing the problem to
a tractable instance of \pbl{Integer Quadratic Programming} as
in Lokshtanov's paper.

\end{abstract}

\section{Introduction}

The crossing number $\crg(G)$ of a graph $G$ is the minimum number of pairwise
edge crossings in a drawing of $G$ in the plane. 
We refer to Section~\ref{sec:prelim} for the definitions of a drawing and
edge crossings.
Finding the crossing number of a graph
is one of the most prominent combinatorial optimization problems in 
graph theory and is NP-hard already in very special cases, e.g.,
even when considering a planar graph with one added
edge~\cite{DBLP:journals/siamcomp/CabelloM13}.
Moreover, we know that computing the crossing number is
APX-hard~\cite{DBLP:journals/dcg/Cabello13}, i.e., there
does not exist a PTAS (unless P=NP).

On the other hand, there is an algorithm~\cite{DBLP:journals/jcss/BhattL84} that
approximates not directly the crossing number, but the quantity $n+\crg(G)$
where $n=|V(G)|$; 
its currently best incarnation does so within a factor of $\cO(\log^2 n)$~\cite{egs}. 
The first sublinear approximation factor of $\tilde{\cO}(n^{0.9})$ 
has been achieved by~\cite{stoccrossing}.
Concerning rectilinear drawings of dense graphs there is another recent
approximation result~\cite{DBLP:journals/comgeo/FoxPS19}.
Much better crossing number approximation results are known for some restricted graph classes,
such as for graphs embeddable in a fixed surface~\cite{ghls,surfaceApprox}
and for graphs from which few edges or vertices can be removed to make them
planar \cite{HS06,cm08,apex,DBLP:journals/jco/ChimaniH17}.

Despite this recent progress in crossing number approximations in special
cases, there are nearly no nontrivial formulas or efficient algorithms 
for computing the exact crossing number in sufficiently ``rich'' graph classes.
Even for very nicely structured classes such as the complete
graphs, the complete bipartite graphs and the toroidal grids (Cartesian
products of cycles), their exact crossing numbers are only conjectured,
but not proved (e.g., we do not know $\crg(K_{13})\,$).

One notable exception (to near-impossibility of computing efficiently 
the exact crossing number) is that the exact crossing number can be
efficiently computed (even in linear time) when it is bounded
\cite{DBLP:journals/jcss/Grohe04,DBLP:conf/stoc/KawarabayashiR07};
more precisely, that \pbl{Crossing Number} is in linear-time 
FPT when parameterized by the solution value.
However, considering {\em nontrivially rich} graph classes with 
unbounded crossing number, there currently seems to be only one such further
efficient result by Biedl et al.~\cite{DBLP:conf/isaac/BiedlCDM17};
computing the exact crossing number for maximal graphs of pathwidth exactly~$3$.

Our paper brings one more small piece to this crossing-number puzzle.
A {\em vertex cover} in a graph $G$ is a set $X$ of vertices of $G$
such that every edge of $G$ has at least one end in~$X$.
We prove that it is possible to compute in FPT the exact crossing number of
a simple graph $G$ when the minimum vertex cover size of $G$ is bounded as a parameter.
\begin{theorem}\label{thm:main}
Given a simple graph $G$,
the problem to compute the crossing number of $G$ and the corresponding
optimal drawing of~$G$ is fixed-parameter tractable with respect to the
parameter~$k=|X|$ where $X\subseteq V(G)$ is a vertex cover of~$G$.
\end{theorem}
\noindent
We remark that computing the minimum vertex cover $X$ is itself in FPT when
parameterized by~$|X|$ \cite{MR1167659}, and so we do not need $X$ on the input.

Although bounding the vertex cover also bounds the pathwidth of a graph,
our Theorem~\ref{thm:main} is incomparable with \cite{DBLP:conf/isaac/BiedlCDM17} since
their result gives exact values only for pathwidth~$3$ (for higher values of
pathwidth \cite{DBLP:conf/isaac/BiedlCDM17} gives an approximation).
Another notable point is that the classes of graphs of bounded vertex cover
are {monotone} (closed under taking subgraphs), while the exact algorithm
in \cite{DBLP:conf/isaac/BiedlCDM17} requires maximal graphs of
pathwidth~$3$ (again, for non-maximal such graphs there is only a
$2$-approximation).

\medskip
In the algorithm of Theorem~\ref{thm:main} we follow the approach of 
Lokshtanov~\cite{DBLP:journals/corr/Lokshtanov15},
who showed that \pbl{Integer Quadratic Programming} is in FPT when
parameterized by the number of variables and the maximum of the 
absolute values of the coefficients.
Lokshtanov then used his IQP algorithm to show that the problem
\pbl{Optimal Linear Arrangement} of a graph $G$ is in FPT when
parameterized by the minimum vertex cover size of~$G$.
In the \pbl{Optimal Linear Arrangement} problem of a graph $G$, the task
is to find a linear ordering of the vertex set of $G$ which minimizes the
sum of ``lengths'' of the edges of $G$.
With respect to this, it is worth to note that the first NP-hardness proof
for \pbl{Crossing Number} \cite{GJ} used a simple reduction from
\pbl{Optimal Linear Arrangement}, one which asymptotically preserves 
almost any reasonable graph parameter including vertex cover.
Although we cannot directly apply our algorithm to the result of that
reduction (due to a presence of parallel edges, as explained below),
a simple modification along the lines of that reduction allows to
deduce Lokshtanov's result for \pbl{Optimal Linear Arrangement}
from our algorithm.

\section{Basic Definitions}\label{sec:prelim}

We use the standard terminology of graph theory.
A special attention has to be paid to simplicity of graphs -- while
(non-)simplicity is usually not an issue for the crossing number
(just subdivide parallel edges), it becomes important with respect to the
minimum vertex cover.
Therefore, we will consider {\em simple graphs} throughout the paper by
default, and we will use the term {\em multigraph} otherwise.

A \emph{drawing} of a graph $G=(V,E)$ is a mapping of the vertices $V$ 
to distinct points in the plane, and of the edges $E$ to simple curves
connecting their respective end points but not containing any other vertex point.
When convenient, we will refer to the elements (vertices and edges) of the
drawing as to the corresponding elements of~$G$.
A \emph{crossing} is a common point of two distinct edge curves, other than their common end point.
It is well established that the search for an optimal solution to the crossing
number problem can be restricted to so called \emph{good drawings}:
any pair of edges crosses at most once, adjacent edges do not cross, and there 
is no crossing point in common of three or more edges.

\begin{definition}\label{def:crnumber}
The problem \pbl{Crossing Number} asks for a {\em good drawing} $D$
of a given graph $G$ with the least possible number of crossings.

The number of crossings in a particular drawing $D$ is denoted by $\crg(D)$ 
and the minimum over all good drawings $D$ of a graph $G$ by~$\crg(G)$.
We call $\crg(D)$ and $\crg(G)$ the {\em crossing number}
of the drawing $D$ and the graph $G$, respectively.
\end{definition}

We will also need to deal with {\em weighted crossing number}.
Consider a graph $H$ with a weight assignment $w:E(H)\to\mathbb N$.
Then a crossing between edges $e,f\in E(H)$ naturally counts as 
$w(e)\cdot w(f)$ crossings (as if they were bunches of $w(e)$ and $w(f)$
parallel edges).
The weighted crossing number $\crg(H)$ of weighted $H$
is defined as in Definition~\ref{def:crnumber} while counting crossings 
this weighted way.

\medskip
Following \cite{DBLP:journals/corr/Lokshtanov15},
we introduce the problem \pbl{Integer Quadratic Programming} (IQP) in a
generalized form.%
\footnote{The stated generalized form comes from page 4 of
\cite{DBLP:journals/corr/Lokshtanov15}, formula (2) and below.}
Its input consists of a $k\times k$ integer matrix~$\mx Q$, an $m\times k$
and $m'\times k$ integer matrices $\mx A$ and $\mx C$,
a $k$-dimensional integer vector $\vec p$,
and $m$- and $m'$-dimensional integer vectors $\vec b$ and $\vec d$. 
The task is to find an optimal solution $\vec z^\circ$ to the following
optimization problem:
\begin{eqnarray}
	\mbox{Minimize\qquad} & \vec z^T\mx Q\vec z &+\> \vec p^T\vec z 
\nonumber\\\label{eq:IQP}
	\mbox{subject to\qquad} & \mx A\vec z &\leq~ \vec b \\
		 & \mx C\vec z &=~ \vec d 
\nonumber\\\nonumber
		 & \vec z &\in~ \mathbb Z^k
\end{eqnarray}
Note that ``finding a solution'' of an IQP instance means exactly one of the
following three outcomes:
the instance is infeasible and we correctly detect that, 
or the instance is feasible and unbounded and we again detect that,
or the instance is feasible and bounded and we output an optimal solution
$\vec z^\circ$.

\begin{theorem}[Lokshtanov \cite{DBLP:journals/corr/Lokshtanov15}]
\label{thm:IQP}
Consider the \pbl{Integer Quadratic Programming} problem as above
\eqref{eq:IQP},
where the input consists of the integer matrices $\mx A$, $\mx C$, $\mx Q$ and
the integer vectors $\vec b$, $\vec d$, $\vec p$.  Let $L$ denote the
length of the combined bit-representation of this input, and let
$\lambda$ be the largest absolute value of the entries in the matrices
$\mx A$, $\mx C$ and $\mx Q$, and the entries in the vector $\vec p$.
There exists an algorithm which finds a solution of this instance of
IQP in time $f(k,\lambda)\cdot L^{\cO(1)}$ for some computable
function~$f$ (that is, fixed-parameter tractable with input size $L$
and parameters $k$ and~$\lambda$).
\end{theorem}

\section{Clustered Optimal Drawings}\label{sec:clustered}

We start with a high-level idea of our solution.
Consider a simple graph $G$ and a vertex cover $X\subseteq V(G)$ of fixed
size~$k=|X|$.
Then $V(G)\setminus X$ is an independent set
and every vertex of $V(G)\setminus X$ can be classified by its neighbourhood
in~$X$ (and this classification is unique up to automorphisms).
At the first sight it thus appears natural to form ``uniform'' clusters of
the vertices with the same neighbourhood (to be treated the same way,
whatever this means), and so seemingly ``reduce'' the input
size to $\cO(2^k)$ and then solve it in FPT by brute force.
This is, unfortunately, not at all sufficient.%
\footnote{%
In the exemplary case of \pbl{Optimal Linear Arrangement}
\cite{DBLP:journals/corr/Lokshtanov15}, one may easily see the problem on the
graph $K_{k,n}$, whose smaller part $X$ of size $k$ is the minimum vertex
cover and all vertices of the larger part form one cluster with the same
neighbourhood~$X$.
Yet, an optimal linear arrangement solution for $K_{k,n}$ has to alternate the vertices
of $X$ and those of the large part in the middle of the arrangement.
Hence in this particular case of OLA, one has to consider at least the relative
position of vertices with respect to $X$ in addition to their neighbourhoods.
}

\medskip
As we will see,
while solving the crossing number problem, it would be enough to additionally 
classify the vertices of $V(G)\setminus X$ by the cyclic ordering of their 
edges in the (yet to be found) optimal drawing of~$G$.
Furthermore, it will also be useful to restrict the arguments
to the aforementioned good drawings
(pairs of edges crosses at most once and adjacent pairs do not cross).
In particular, in a good drawing the edges incident to one common vertex
always form an uncrossed star.
We give the following core definition (see also Figure~\ref{fig:clusters}):
\begin{definition}\label{def:clustering}
Let $G$ be a graph with a vertex cover $X$, and $D$ be a good drawing of $G$.
Then two vertices $x,y\in V(G)\setminus X$ belong to the same {\em
topological cluster} in $D$ (with implicit respect to~$X$)
if their neighbourhood in $X$ is the same, and
the clockwise cyclic order of the neighbours of $x$ within $D$ is the same
as the clockwise cyclic order of the neighbours of $y$.
\\(Note that a vertex of $X$ does {\em not} belong to any topological cluster in~$D$.)
\end{definition}

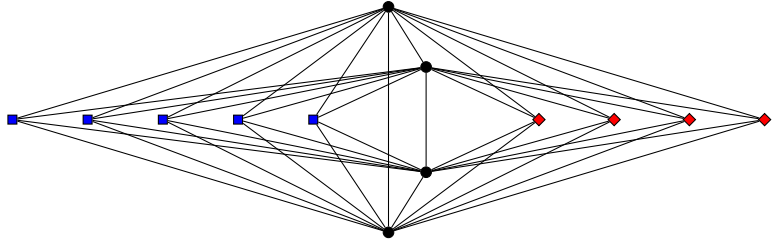
\begin{figure}[t]
\centering
\begin{tikzpicture}[
        triangle/.style = {regular polygon, regular polygon sides=3},
        square/.style = {regular polygon, regular polygon sides=4},
]
\tikzstyle{every node}=[draw, fill=black, shape=circle, minimum size=4pt, inner sep=1.2pt];
\node (m1) at (0,0) {};
\node (m2) at (0,3) {};
\node (m3) at (0.5,0.8) {};
\node (m4) at (0.5,2.2) {};
\node[fill=blue,square] (l1) at (-5,1.5) {};
\node[fill=blue,square] (l2) at (-4,1.5) {};
\node[fill=blue,square] (l3) at (-3,1.5) {};
\node[fill=blue,square] (l4) at (-2,1.5) {};
\node[fill=blue,square] (l5) at (-1,1.5) {};
\node[fill=red,diamond] (r1) at (5,1.5) {};
\node[fill=red,diamond] (r2) at (4,1.5) {};
\node[fill=red,diamond] (r3) at (3,1.5) {};
\node[fill=red,diamond] (r4) at (2,1.5) {};
\draw (m1)--(m3)--(m4)--(m2)--(m1);
\draw (m1)--(l1)--(m2)--(r1)--(m1);
\draw (m1)--(l2)--(m2)--(r2)--(m1);
\draw (m1)--(l3)--(m2)--(r3)--(m1);
\draw (m1)--(l4)--(m2)--(r4)--(m1);
\draw (m1)--(l5)--(m2);
\draw (m4)--(l1)--(m3)--(r1)--(m4);
\draw (m4)--(l2)--(m3)--(r2)--(m4);
\draw (m4)--(l3)--(m3)--(r3)--(m4);
\draw (m4)--(l4)--(m3)--(r4)--(m4);
\draw (m4)--(l5)--(m3);
\end{tikzpicture}
\caption{An illustration of topological clusters with respect to the vertex
cover formed by the middle four black vertices:
any two blue (square) vertices on the left belong to the same topological cluster,
and likewise for any two red (diamond) vertices.
Any blue and a red vertex belong to different topological clusters.
}
\label{fig:clusters}
\end{figure}

We aim to show that an optimal drawing of our graph $G$ can be obtained in
such a form that the topological clusters of vertices and the incident edges
are ``drawn closely together''.
The same idea, only for the special case of the complete bipartite graph
$G=K_{k,n}$, has already been used by Christian, Richter and
Salazar~\cite{DBLP:journals/jct/ChristianRS13} (their research goal, though,
was different).
Our paper can be considered a generalization of (a part 
of)~\cite{DBLP:journals/jct/ChristianRS13}.
To achieve our goal, we separate two kinds of crossings and rigorously describe
a topological clustering of a drawing~of~$G$.

Assume a good drawing $D$ of~$G$.
Having two edges $e,f\in E(G)$ such that one end of $e$ belongs to the same
topological cluster in $D$ as one end of~$f$,
we say that a (possible) crossing of $e$ and $f$ is a {\em cluster crossing}.
All other edge crossings occurring in $D$ are called {\em non-cluster crossings}
(and they include all crossings on edges with both ends in~$X$).
Here we denote by $\crgnc(D)$ the number of non-cluster crossings
(possibly weighted) in the drawing~$D$.

In the following definition we, informally, select one weighted
``representative'' of each topological cluster of a drawing~$D$.
\begin{definition}\label{def:topocluster}
A drawing $D_X$ is called a {\em topological clustering} of the drawing $D$
(of a graph $G$) with respect to its vertex cover $X$ if the following hold:
\begin{itemize}\vspace*{-1ex}
\item $D_X$ is an induced subdrawing of $D$ and $V(D_X)\supseteq X$,
\item every vertex of $V(D)\setminus X$ belongs to the same topological cluster
in $D$ as some vertex of $V(D_X)\setminus X$,
\item no two vertices of $V(D_X)\setminus X$ are in the same topological
cluster in $D_X$,
and \item $D_X$ is equipped with a weight function 
$c:V(D_X)\setminus X\to\mathbb N$ such that, for every $t\in V(D_X)\setminus X$,
the size of the topological cluster in $D$
containing $t$ equals~$c(t)$.
\end{itemize}
\end{definition}
\noindent
Note that there can be many (topologically) different
topological clusterings $D_X$ of the same drawing $D$, 
depending on how the ``representative'' vertices from $V(D_X)\setminus X$ 
are chosen.
See also Figure~\ref{fig:clusterings}

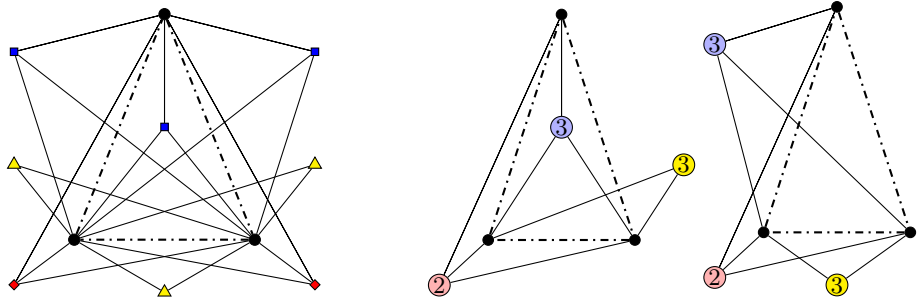
\begin{figure}[t]
\centering
\begin{tikzpicture}[xscale=0.8,
        triangle/.style = {regular polygon, regular polygon sides=3},
        square/.style = {regular polygon, regular polygon sides=4},
]
\tikzstyle{every node}=[draw, fill=black, shape=circle, minimum size=4pt, inner sep=1pt];
\node[inner sep=1.5pt] (m1) at (0,0) {};
\node[inner sep=1.5pt] (m2) at (1.5,3) {};
\node[inner sep=1.5pt] (m3) at (3,0) {};
\node[fill=blue,square] (l1) at (1.5,1.5) {};
\node[fill=blue,square] (l2) at (-1,2.5) {};
\node[fill=blue,square] (l3) at (4,2.5) {};
\node[fill=red,diamond] (r1) at (-1,-0.6) {};
\node[fill=red,diamond] (r2) at (4,-0.6) {};
\node[fill=yellow,triangle] (s1) at (1.5,-0.7) {};
\node[fill=yellow,triangle] (s2) at (-1,1) {};
\node[fill=yellow,triangle] (s3) at (4,1) {};
\draw[thick, dash dot] (m1)--(m2)--(m3)--(m1);
\draw (m1)--(l1)--(m2)--(l1)--(m3);
\draw (m1)--(l2)--(m2)--(l2)--(m3);
\draw (m1)--(l3)--(m2)--(l3)--(m3);
\draw (m1)--(r1)--(m2)--(r1)--(m3);
\draw (m1)--(r2)--(m2)--(r2)--(m3);
\draw (m1)--(s1)--(m3);
\draw (m1)--(s2)--(m3);
\draw (m1)--(s3)--(m3);
\end{tikzpicture}
\qquad\qquad
\begin{tikzpicture}[xscale=0.65]
\footnotesize
\tikzstyle{every node}=[draw, fill=black, shape=circle, minimum size=4pt, inner sep=0.3pt];
\node[inner sep=1.2pt] (m1) at (0,0) {};
\node[inner sep=1.2pt] (m2) at (1.5,3) {};
\node[inner sep=1.2pt] (m3) at (3,0) {};
\node[fill=blue!30!white] (l1) at (1.5,1.5) {3};
\node[fill=red!30!white] (r1) at (-1,-0.6) {2};
\node[fill=yellow] (s3) at (4,1) {3};
\draw[thick, dash dot] (m1)--(m2)--(m3)--(m1);
\draw (m1)--(l1)--(m2)--(l1)--(m3);
\draw (m1)--(r1)--(m2)--(r1)--(m3);
\draw (m1)--(s3)--(m3);
\end{tikzpicture}
\begin{tikzpicture}[xscale=0.65]
\footnotesize
\tikzstyle{every node}=[draw, fill=black, shape=circle, minimum size=4pt, inner sep=0.3pt];
\node[inner sep=1.2pt] (m1) at (0,0) {};
\node[inner sep=1.2pt] (m2) at (1.5,3) {};
\node[inner sep=1.2pt] (m3) at (3,0) {};
\node[fill=blue!30!white] (l1) at (-1,2.5) {3};
\node[fill=red!30!white] (r1) at (-1,-0.6) {2};
\node[fill=yellow] (s3) at (1.5,-0.7) {3};
\draw[thick, dash dot] (m1)--(m2)--(m3)--(m1);
\draw (m1)--(l1)--(m2)--(l1)--(m3);
\draw (m1)--(r1)--(m2)--(r1)--(m3);
\draw (m1)--(s3)--(m3);
\end{tikzpicture}
\caption{An illustration: the graph on the left has a vertex cover $X$ formed by
the three black vertices.
There are three topological clusters wrt.~$X$, depicted by the blue
(square), red (diamond) and yellow (triangle) vertices.
On the right, we can see two different topological clusterings of this
graph, with the weights in the node circles.
}
\label{fig:clusterings}
\end{figure}

In view of Definition~\ref{def:topocluster},
it will be useful to consider the crossing number of the following 
{\em independently weighted} graphs.
For a graph $H$, an independent set $Y\subseteq V(H)$,
and a weight function $c:Y\to\mathbb N$,
let the edge weights in $H$ be as follows:
for $e\in E(H)$ having one end $x\in Y$ we set $c'(e):=c(x)$, 
and for edges $e$ with both ends in $V(H)\setminus Y$ we set~$c'(e)=1$.
This defines the weighted crossing number of $H$ and,
in particular for $Y=V(D_X)\setminus X$, 
the weighted crossing number $\crg(D_X)$
of $D_X$ with respect to its weight function~$c$.

We can now formulate and prove the core claim:

\begin{lemma}\label{lem:noncluster}
For every good drawing $D$ of a graph $G$ with a vertex cover $X$,
there exists its topological clustering $D_X$ such that the number of
non-cluster crossings in $D$ is at least $\crg(D_X)$.
\end{lemma}

\begin{proof}
We start by setting $D':=D$, and then
we will inductively choose suitable representatives of
the topological clusters of $D'$ until we arrive at desired~$D_X$.

Let the {\em cost} of a vertex $x$ of $D'$ be the sum of all
non-cluster(!) crossings carried by the edges of $D'$ incident to~$x$.
We pick any (nonempty) topological cluster $S\subseteq V(D')\setminus X$
and choose a vertex $s_0\in S$ having the least cost among those of~$S$.
If $|S|>1$, we iterate the following over all $s\in S\setminus\{s_0\}$;
\begin{itemize}\vspace*{-1ex}
\item remove $s$ and its incident edges from the drawing $D'$,
\item choose a new vertex $s'$ in a tiny neighbourhood of $s_0$ in $D'$,
and draw the edges of $s'$ in a tiny strip along the edges of $s_0$ 
while making only such non-cluster crossings as those that exist 
on the edges of $s_0$.
\end{itemize}
Since $s_0$ and $s$ have had the same neighbourhood, the new drawing $D''$
as a graph
(with~$s'$) is isomorphic to $D'$ and the vertices $s_0$ and $s'$ still belong
to the same topological cluster.
Moreover, since $s_0$ has been chosen with the least cost,
we have $\crgnc(D'')\leq\crgnc(D')$.
(New cluster crossings can be simply ignored.)

At the end of the previous procedure, we get a drawing $D^\circ$
(isomorphic to~$D'$) which has no more non-cluster crossings than~$D'$,
and all edges of the cluster $S$ are drawn ``the same way closely together''
in~$D^\circ$.
From this it follows that the number of non-cluster crossings carried by
the edges incident to the cluster $S$ in $D^\circ$ equals $|S|$-times 
this number on the edges incident to $s_0$ in $D'$.
We therefore define $D_1$ as the induced subdrawing of $D'$ obtained by
deleting the vertices of $S\setminus\{s_0\}$ and assigning the weight $c(s_0)=|S|$.
In the setting of independently weighted graph underlying $D_1$,
we get $\crgnc(D_1)=\crgnc(D^\circ)\leq\crgnc(D')$.

\smallskip
We finish the proof inductively.
Let $r$ be the number of topological clusters of given $D$.
For $i=1,2,\ldots,r-1$, we now repeat the previous steps
for $D':=D_i$, obtaining a subdrawing $D_{i+1}$ such that
$\crgnc(D_{i+1})\leq\crgnc(D_i)$.
Finally, $D_{r}$ is a topological clustering of~$D$ with respect to $X$
by Definition~\ref{def:clustering}
and we conclude $\crg(D_{r})=\crgnc(D_{r})\leq\ldots\leq\crgnc(D_1)\leq\crgnc(D)$.
\qed\end{proof}

\section{Counting the Crossings in Clusters and Between}\label{sec:counting}

In order to complement Lemma~\ref{lem:noncluster},
we need to estimate also the number of cluster crossings in a drawing~$D$.
This is actually quite easy using the fact that two vertices in the same
topological cluster have the same cyclic ordering of their neighbours.
We use the following simple claim (cf.~Figure~\ref{fig:clusterstacking}):

\begin{lemma}[{\cite[Lemma~2.1]{DBLP:journals/jct/ChristianRS13}}]
\label{lem:K2mcr}
Let $x,y$ be the two vertices of degree $m$ in $K_{2,m}$ for~$m\geq3$.
Consider any good drawing $D$ of $K_{2,m}$ such that the clockwise cyclic
order of the neighbours of $x$ within $D$ is the same as the 
clockwise cyclic order of the neighbours of~$y$.
Then $\crg(D)\geq\lfloor\frac m2\rfloor\cdot\lfloor\frac{m-1}2\rfloor
 :=Z(m)$.
\end{lemma}

\begin{corollary}\label{cor:clustercr}
Consider a good drawing $D$ of a graph $G$ with a vertex cover~$X$,
and a topological cluster $S\subseteq V(D)\setminus X$ of size~$c=|S|$.
Let the degree of vertices in $S$ be~$m$.
Then the number of cluster crossings in $D$ between the edges 
incident with $S$ is at least
${c\choose2}\cdot\lfloor\frac m2\rfloor\cdot\lfloor\frac{m-1}2\rfloor
 ={c\choose2}\cdot Z(m)$.
\end{corollary}
Readers may notice that the formula 
${c\choose2}\cdot\lfloor\frac m2\rfloor\cdot\lfloor\frac{m-1}2\rfloor$
in the lemma is not symmetric in $c$ and $m$ -- 
it grows on one hand with $c^2/2$ and on the other hand with $m^2/4$.
This is correct since the setting is also not symmetric.
The vertices in $S$ are required to have the same cyclic order of
neighbours in~$D$, but the neighbours of $S$ do not have this property.

\begin{proof}
Let $Z(m)=\lfloor\frac m2\rfloor\cdot\lfloor\frac{m-1}2\rfloor$.
For $s\in S$, denote by $R_s\subseteq G$ the subgraph
(a~star) induced by $s$ and the incident edges of $s$ in~$G$.

There is nothing to prove (the bound equals~$0$) for $m\leq2$ or $c=1$.
Otherwise, for every pair $s_1,s_2\in S$, $s_1\not=s_2$,
we apply Lemma~\ref{lem:K2mcr} to the subdrawing $D_{s_1,s_2}$ of $D$ 
induced by $R_{s_1}\cup R_{s_2}$, getting at least $Z(m)$ crossings
within $D_{s_1,s_2}$.
If $s_3\in S$ is different from $s_1$ and $s_2$, then the crossings
in $D_{s_1,s_3}$ are all distinct from the crossings in $D_{s_1,s_2}$;
this is since $E(D_{s_1,s_2})\cap E(D_{s_1,s_3})=E(R_{s_1})$, but
the edges on $R_{s_1}$ are all incident to $s_1$ and so they cannot
mutually cross in a good drawing.
Consequently, each of the $c\choose2$ invocations of Lemma~\ref{lem:K2mcr}
contributes a collection of at least $Z(m)$ new crossings,
providing the overall lower bound of ${c\choose2}\cdot Z(m)$
cluster crossings between the edges incident with $S$.
\qed\end{proof}

\begin{figure}[t]
\centering
\begin{tikzpicture}[xscale=0.9]
\tikzstyle{every node}=[draw, fill=black, shape=circle, minimum size=4pt, inner sep=0.3pt];
\node (l1) at (-3,0) {};
\node (l2) at (-2,0) {};
\node (l3) at (-1,0) {};
\node (l5) at (1,0) {};
\node (l6) at (2,0) {};
\node (l7) at (3,0) {};
\node (l8) at (4,0) {};
\node[fill=white] (m1) at (0,1) {};
\node[fill=white] (m2) at (0,2) {};
\draw (m1)--(l1)--(m2);
\draw (m1)--(l2)--(m2);
\draw (m1)--(l3)--(m2);
\draw (m1)--(l5)--(m2);
\draw (m1)--(l6)--(m2);
\draw (m1)--(l7)--(m2);
\draw (m1)--(l8)--(m2);
\end{tikzpicture}
\qquad\qquad
\begin{tikzpicture}[xscale=0.6,yscale=0.8]
\tikzstyle{every node}=[draw, fill=black, shape=circle, minimum size=2.5pt, inner sep=0.3pt];
\node (l1) at (-3,0) {};
\node (l2) at (-2,0) {};
\node (l3) at (-1,0) {};
\node (l5) at (1,0) {};
\node (l6) at (2,0) {};
\node (l7) at (3,0) {};
\node (l8) at (4,0) {};
\node[fill=white] (m1) at (0,0.5) {};
\node[fill=white] (m2) at (0,1) {};
\node[fill=white] (m3) at (0,1.5) {};
\node[fill=white] (m4) at (0,2) {};
\node[draw=none,fill=none] (mx) at (0,2.7) {$\vdots$};
\node[fill=white] (m5) at (0,3.5) {};
\node[fill=white] (m6) at (0,4.1) {};
\tikzstyle{every path}=[very thin];
\draw (m1)--(l1)--(m2);\draw (m1)--(l2)--(m2);\draw (m1)--(l3)--(m2);
\draw (m1)--(l5)--(m2);\draw (m1)--(l6)--(m2);\draw (m1)--(l7)--(m2);
\draw (m1)--(l8)--(m2);
\draw (m3)--(l1)--(m4);\draw (m3)--(l2)--(m4);\draw (m3)--(l3)--(m4);
\draw (m3)--(l5)--(m4);\draw (m3)--(l6)--(m4);\draw (m3)--(l7)--(m4);
\draw (m3)--(l8)--(m4);
\draw (m5)--(l1)--(m6);\draw (m5)--(l2)--(m6);\draw (m5)--(l3)--(m6);
\draw (m5)--(l5)--(m6);\draw (m5)--(l6)--(m6);\draw (m5)--(l7)--(m6);
\draw (m5)--(l8)--(m6);
\end{tikzpicture}
\caption{Left: an optimal drawing of $K_{2,7}$ achieving the minimum number
of $Z(7)=9$ crossings among all drawings in which the two vertices of degree
$7$ have the same cyclic order of their neighbours (as in Lemma~\ref{lem:K2mcr}).
\protect\\
Right: ``stacking'' the left subdrawings such that the total number of
cluster crossings here matches the lower bound given by Corollary~\ref{cor:clustercr}.}
\label{fig:clusterstacking}
\end{figure}
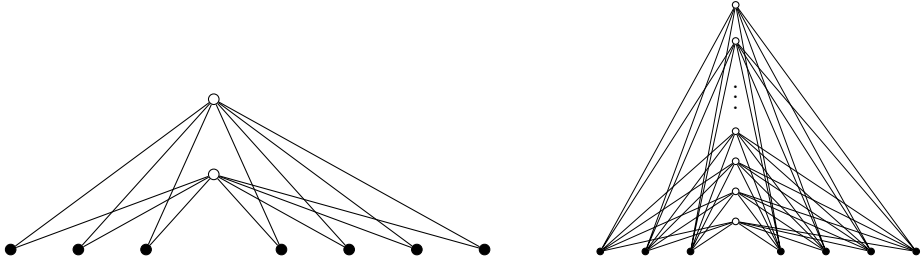

The next step is to introduce an ``abstract level'' of a topological clustering.
Simply put, a drawing $D$ is an {\em abstract topological clustering}  of
a graph $G$ with respect to its vertex cover~$X$
if $D$ is a topological clustering of {\em some} drawing of $G$, but without
the weight function. More precisely:

\begin{definition}\label{def:abstratcluster}
A drawing $C_X$ is an {\em abstract topological clustering} of a graph $G$
with respect to its vertex cover $X$ if the following hold:
\begin{itemize}\vspace*{-1ex}
\item $C_X$ is a good drawing of an induced subgraph of $G$ containing~$X$,
\item for every vertex $w\in V(G)\setminus X$ there is a vertex in 
$V(C_X)\setminus X$ having the same neighbourhood as~$w$ in~$G$,
and
\item no two vertices of $V(C_X)\setminus X$ are in the same topological
cluster in $C_X$.
\end{itemize}
\end{definition}

We will further use the term of {\em planarization} of a drawing $D$,
which is the plane graph obtained from $D$ by turning every crossing into a
new degree-$4$ vertex.
Two drawings $D_1$ and $D_2$ of the same graph are {\em combinatorially equivalent}
if the same pairs of edges cross in $D_1$ as in $D_2$, moreover in the same
order of the crossings on each edge,
and their planarizations are equivalent plane graphs
(i.e., with the same collection of faces).

\begin{lemma}\label{lem:abstractbd}
Consider a graph $G$ with a vertex cover~$X$ of size $k=|X|$.
Then every abstract topological clustering of $G$ has size at most
singly exponential in~$k$,
and the number of combinatorially non-equivalent abstract topological
clusterings of $G$ is bounded from above by a doubly exponential function
of~$k$.
\end{lemma}

\begin{proof}
A topological cluster in a drawing of $G$ is uniquely determined by one of
$2^k$ possible neighbourhood subsets in $X$, and one of up to $(k-1)!$ cyclic
orders of the neighbours.
Hence an abstract topological clustering $C$ of $G$ has at most
$k+2^k(k-1)!\leq k^{\cO(k)}$ vertices.
Hence the number of edges and of pairs of edges of $C$ is bounded from
above by $k^{\cO(k)}$, which also implies $\crg(C)\leq k^{\cO(k)}$.
Hence the planarization of $C$ has at most $k^{\cO(k)}$ vertices,
and there are altogether at most $2^{k^{\cO(k)}}$ such possible
nonequivalent planarizations of abstract topological clusterings of~$G$.
\qed\end{proof}

A consequence of Lemma~\ref{lem:abstractbd} is that we can, in FPT time,
process all possible abstract topological clusterings of any graph $G$ 
with a small vertex cover.
Therefore, from now on, we may just fix one abstract topological clustering
$C_X$ of~$G$ and discuss how to optimize the crossing number over all
such drawings of $G$ whose topological clustering comes from~$C_X$.
The latter problem will be reduced to a bounded instance of IQP, similarly
as the special case of complete bipartite graphs has been handled
in aforementioned \cite{DBLP:journals/jct/ChristianRS13}.

\subsection{IQP Formulation for Crossings}

In regard of Definition~\ref{def:abstratcluster} and the coming arguments,
it will be useful to consider the following ``compressed'' representation of
a graph $G$ with a small vertex cover~$X$.
Let $G_X$ denote the subgraph of $G$ induced by~$X$, and
consider the function $h:2^X\to\mathbb N_0$ such that,
for any $Y\subseteq X$,
$h(Y)$ is the number of vertices of $G$ outside of $X$ 
whose neighbourhood in $G$ is exactly~$Y$.
Clearly, $G_X$ and $h$ determine $G$ up to an isomorphism
(and the size of this description can be only logarithmic compared to the
size of~$G$).

For given $G$ and $X$, let us fix any abstract topological
clustering $C_X$ of $G$ with respect to $X$.
For $Y \subseteq X$, let $S(Y)$
be the set of vertices of $V(C_X) \setminus X$ whose neighborhood in
$X$ is exactly $Y$. 
Note that $h(Y)$ is non-zero iff $S(Y)$ is non-empty.
Let $Y_1, \ldots, Y_l$
be an enumeration of all subsets of $X$ which map to a non-zero value
under $h$; then $\bigcup_{i = 1}^{\,l} S(Y_i) = V(C_X) \setminus X$.
For $i \in \{1, \ldots, l\}$, let
$g(i) = |S(Y_i)|$ and let $S(Y_i) = \{v_{(i, 1)}, \ldots, v_{(i,
  g(i))}\}$.

For an illustration, in Figure~\ref{fig:clusterings}
(on the right, but ignoring the weights since we are considering
an abstract clustering) we have got $l=2$ ($Y_1=X$ and $Y_2$ are the
two bottom vertices of $X$), and
$g(1)=2$ (blue and red clusters) and $g(2)=1$ (yellow cluster).
Altogether, $V(C_X) \setminus X$ has three vertices there.

Let $I$ be an index set defined as $I := \{ (i, j) \mid 1 \leq i \leq
l, 1 \leq j \leq g(i)\}$.  
Similarly as in~\cite{DBLP:journals/jct/ChristianRS13},
we define the following \emph{crossing vector} 
$\vec{p} = (p_\alpha \mid \alpha\in I)$ and the
\emph{crossing matrix} $\vec{Q} = (q_{\alpha,
  \beta} \mid \alpha,\beta\in I)$,
such that the intended use of $\vec p$ is to count the 
crossings between the edges of $G_X$ 
and the edges incident to each topological cluster
corresponding to a vertex of $V(C_X)\setminus X$,
and the intended use of $\mx Q$ is to count the cluster crossings
of each one of the topological clusters (the diagonal entries)
and the non-cluster crossings between pairs of the clusters
(the other entries):
\begin{itemize}\vspace*{-1ex}
\item The crossing vector $\vec{p}\,$: 
  Let $e_1, \ldots, e_r$ be an enumeration of the edges in $G_X$.
  For $\alpha\in I$ and $i \in \{1, \ldots, r\}$,
  let $p^i_\alpha$ be the number of edges incident to $v_\alpha$ that
  cross $e_i$ in $C_X$. Then $p_\alpha = \sum_{i = 1}^{r}
  p^i_\alpha$.
\item The crossing matrix $\vec{Q}\,$:
  Let $\alpha,\beta\in I$.
  If $\alpha \neq \beta$, define
  $q_{\alpha, \beta}$ as the number of crossings in $C_X$ between the
  edges incident to $v_\alpha$ and the edges incident to~$v_\beta$. If
  $\alpha = \beta = (i, j)$, then $q_{\alpha, \alpha} :=
  Z(|Y_i|) = \big\lfloor  \frac{|Y_i|}{2} \big\rfloor
  \cdot \big\lfloor \frac{|Y_i|-1}{2} \big\rfloor$.
\end{itemize}
($Z(\cdot)$ has been defined in Lemma~\ref{lem:K2mcr} and
Corollary~\ref{cor:clustercr}.)

To recapitulate where we stand now;
we have fixed an abstract topological clustering $C_X$ of $G$,
and in order to proceed to a drawing of $G$ (underlied by~$C_X$),
we first need to assign suitable integer weights to the vertices of 
$V(C_X) \setminus X$.
Our goal is to minimize the total number of crossings in 
the constructed drawing of $G$.
However, we only have $C_X$ and some assigned weights on $V(C_X) \setminus X$, which
together define the topological clustering $D_X$ of a drawing of~$G$.
The crossing number $\crg(D_X)$ is, via Lemma~\ref{lem:noncluster},
related to the number of non-cluster crossings in a desired drawing of~$G$
(refer to the proof of Theorem~\ref{thm:mainbits} for a precise formulation).
But it is not sufficient to minimize $\crg(D_X)$ since the cluster crossings
in a drawing of $G$ also play important role.

To complete the picture with cluster crossings, 
we define (cf.~Corollary~\ref{cor:clustercr})
$$
\clg(D_X) := \sum_{t\in V(D_X)\setminus X}
  {c(t)\choose2}\cdot Z\big(d(t)\big)
 = \sum_{t\in V(D_X)\setminus X}
  {c(t)\choose2}\cdot\left\lfloor\frac{d(t)}2\right\rfloor\cdot
        \left\lfloor\frac{d(t)-1}2\right\rfloor
,$$
where $c$ is the weight function of~$D_X$ and
$d(t)$ denotes the degree of $t$ (which is the same in $D_X$ as in~$G$).
Again, we refer to the proof of Theorem~\ref{thm:mainbits} for further
treatment of the relation of $\clg(D_X)$ to the cluster crossings
in a drawing of~$G$.

\begin{lemma}\label{lem:IQPcross}
Let $C_X$ be an abstract topological
clustering of $G$ with respect to a vertex cover $X$,
and denote by $\mathcal{D}(C_X)$ the set of all topological clusterings of
good drawings of~$G$ whose unweighted topological clustering is $C_X$.
Let, furthermore, $Y_i$, $g$, $I$, $\vec p$ and $\mx Q$ be as above.
Then the following IQP
\begin{eqnarray}
	\mbox{\rm Minimize\qquad} & f(\vec{z}) & = 
	\vec{z}^T \mx{Q} \vec{z} \>+\> 2\cdot\vec{p}^T \vec{z} 
\label{eq:IPQcross}\\\nonumber
        \mbox{\rm over all\qquad} & \vec{z} & =
	 \big(z_{(1, 1)}, \ldots, z_{(1, g(1))},
        \ldots, z_{(l, 1)}, \ldots, z_{(l, g(l))}\big)
\\\nonumber
        \mbox{\rm subject to\qquad} &   \sum\limits_{j = 1}^{g(i)}
	z_{(i, j)}  & = h(Y_i) ~~~~~~~\mbox{for}~ i \in \{1, \ldots, l\}
\\\nonumber
        & z_{(i, j)} & \ge 0 ~~~~~~~~~~~~\mbox{for}~ (i, j) \in I
\\\nonumber
        & \vec z & \in \,\mathbb{Z}^{|I|}
\end{eqnarray}
computes the minimum value of $2\cdot(\crg(D)+\clg(D)-r)$ over all 
$D\in \mathcal{D}(C_X)$,
where $r=\crg(C_X|_X)$ is the number of crossings in the subdrawing of $C_X$
induced by the vertex set~$X$.
\end{lemma}

\begin{proof}
First, note that for any $D\in \mathcal{D}(C_X)$ we have $\crg(D|_X)=r$ by
definition.
For a particular weight assignment~$\vec z$, consider the corresponding
topological clustering $D=D(C_X,\vec z)\in \mathcal{D}(C_X)$.
We write $\crg(D)=r+r_1(D)+r_2(D)$ where $r_1(D)$ counts the
(weighted) crossings in $D$ which involve one edge with both ends in $X$,
and $r_2(D)$ counts the crossings of which neither edge has both ends
in~$X$.
From the definition of the crossing vector $\vec p$ we immediately have
$r_1(D)=\vec{p}^T \vec{z}$.
From the definition of the crossing matrix $\vec Q$ and that of
$\clg(\cdot)$ we also get $r_2(D)+\clg(D)=\frac12\cdot \vec{z}^T
\mx{Q}\vec{z}$.
Altogether, $\frac12f(\vec z)=r_1(D)+r_2(D)+\clg(D)=\crg(D)+\clg(D)-r$.
\qed\end{proof}

We are now ready to prove the main result of this paper, which is
as stated below.

\begin{theorem}[refinement of Theorem~\ref{thm:main}]\label{thm:mainbits}
Consider a simple graph $G$ given on the input as follows:
there is a set $X$ (a vertex cover of~$G$), a simple graph $G_X$
(which is the subgraph of $G$ induced by~$X$), and a function
$h:2^X\to\mathbb N_0$ such that, for $Y\subseteq X$,
$h(Y)$ is the number of vertices of $G$ outside of $X$ 
whose neighbourhood in $G$ is exactly~$Y$.
The size of this input $G$ equals the size of $G_X$ plus the the length of
the bit-representation of function~$h$.

Then the problem to compute the crossing number of $G$ and the corresponding
topological clustering of an optimal drawing of~$G$ 
is fixed-parameter tractable with respect to the parameter~$k=|X|$.
\end{theorem}

Note that, when the vertex cover size $k=|X|$ is fixed,
the size of the input $G$ described in Theorem~\ref{thm:mainbits} is 
logarithmic in the number of vertices of~$G$.
Although, in a typical use case, in which we do not get the input graph $G$
in a parsed form as in Theorem~\ref{thm:mainbits}, but rather as a list of vertices
and edges, we can first compute, again in FPT \cite{MR1167659}, 
a vertex cover~$X$ of size $\leq k$ and the corresponding function~$h$.
Then, from the output topological clustering of an optimal drawing of~$G$,
we can easily in polynomial time construct the corresponding drawing of~$G$.
Hence Theorem~\ref{thm:mainbits} implies Theorem~\ref{thm:main}.

\begin{proof}
Consider an optimal drawing $D_0$ of $G$, i.e., one for which
$\crg(D_0)=\crg(G)$ holds.
Then $D_0$ may be assumed a good drawing by folklore arguments.
By Lemma~\ref{lem:noncluster}, there is a topological clustering $D_X$
of $D_0$ such that $\crgnc(D_0)\geq\crg(D_X)$.
Recall that $D_X$ is equipped with the weight function~$c$, and that
$\clg(D_X) =$ $\sum_{t\in V(D_X)\setminus X}
  {c(t)\choose2}\cdot Z\big(d(t)\big)$ 
where $d(t)$ denote the degree of $t$.
By Corollary~\ref{cor:clustercr}, the total number of cluster crossings 
in $D_0$ is at least~$\clg(D_X)$.

Now, let $C_X$ be the abstract topological clustering underlying~$D_X$.
Although we do not (yet) know $D_X$, we can ``find'' $C_X$ by a brute force
enumeration of all abstract topological clusterings of~$G$, which is still
in FPT by Lemma~\ref{lem:abstractbd}.
Precisely, for every possible~$C_X$ (where ``possible'' is checked
simply by brute force with respect to the parameter~$k$),
we compose an IQP as above \eqref{eq:IPQcross}.
Then, using Theorem~\ref{thm:IQP}, we solve it to get
an assignment $\vec z$ of weights to $C_X$, leading to a clustering $D_X'$,
such that the objective value $\clg(D_X')+\crg(D_X')$ is minimized
over all $D_X'\in\mathcal{D}(C_X)$ by Lemma~\ref{lem:IQPcross}.
Let, furthermore, $C_X^\circ$ be an abstract topological clustering of~$G$
achieving the overall minimum value of the IQP solutions --
this leads to a clustering $D^\circ_X$ with globally minimal
$\clg(D_X^\circ)+\crg(D_X^\circ)$ for given~$G$.

Consequently, counting separately the cluster and non-cluster crossings in~$D_0$,
and then considering the minimality of $D^\circ_X$, we get
$$
\crg(D_0)\geq \clg(D_X)+\crgnc(D_0)\geq
 \clg(D_X)+\crg(D_X)\geq \clg(D_X^\circ)+\crg(D_X^\circ)
.$$
It is now enough to ``lift'' the clustering $D_X^\circ$ into a corresponding
drawing $D_1$ of $G$ with $\clg(D_X^\circ)+\crg(D_X^\circ)$ crossings,
which follows straightforwardly in the same way as in 
\cite{DBLP:journals/jct/ChristianRS13}, see Figure~\ref{fig:clusterstacking}.
Hence $\clg(D_X^\circ)+\crg(D_X^\circ)\geq\crg(G)=\crg(D_0)$,
and so $\crg(D_1)=\crg(D_0)=\crg(G)$.

It remains to address runtime of our procedure.
In the IQP \eqref{eq:IPQcross} we have $|I|$ bounded from above by the size
of $C_X$, which is at most singly exponential in~$k$ by
Lemma~\ref{lem:abstractbd}.
The same asymptotic upper bound $k^{\cO(k)}$ from the proof of
Lemma~\ref{lem:abstractbd} applies also to $\crg(C_X)$, and this clearly
bounds all the entries of the matrix $\mx Q$ and the vector~$\vec p$.
Let $L$ be the length of the bit representation of $h$
(from the input representation of~$G$); then the length of the combined bit
representation of the IQP \eqref{eq:IPQcross} is at most
$f_1(k)\cdot L^{\cO(1)}$ for some computable (singly exponential) function $f_1$.
Then from Theorem~\ref{thm:IQP}, \eqref{eq:IPQcross} is solved by an algorithm 
in FPT time $f_2(k)\cdot L^{\cO(1)}$ for some computable function~$f_2$.
This IQP step is repeated, by brute force and independently
of $L$, at most $f_3(k)$ times
where $f_3$ is a computable function (doubly exponential) coming from the
bound on the number of abstract clusterings in Lemma~\ref{lem:abstractbd}.
\qed\end{proof}

\section{Conclusions}\label{sec:conclu}

In our work we have stressed simplicity of the considered graphs.
A natural question is about what happens if we consider multigraphs with a
vertex cover of size~$k$.
There is, unfortunately, no easy answer to this question since deep
problems arise in two different places of our arguments.
First, since the multiplicity of an edge may be unbounded in~$k$,
the entries of the crossing vector $\vec p$ and the crossing matrix $\mx Q$
would no longer be bounded in~$k$.
Second, when defining topological clusters, it would no longer be enough to
consider a bounded number of neighbourhoods in $X$ and a bounded number of
cyclic orders, but also a potentially unbounded number of different
multiplicities of the edges in a cluster.
Each of these problems would completely ruin the runtime of our procedure.

Therefore, we leave the problem of computational complexity of the exact
crossing number of multigraphs parameterized by a vertex cover size as open,
for future research.
On the other hand, in the special case of multigraphs with a vertex cover of
size $k$ and edge multiplicities bounded by a computable function of~$k$,
it is not difficult to extend our approach to obtain again an FPT algorithm
(which we skip here due to space restrictions).
 
At last, we would like to very briefly mention the problem of minimizing
the crossing number of a small perturbation of a given map of a graph,
e.g.~\cite{DBLP:conf/gd/FulekT18}, which shares some common ground with our
arguments.
Although the objectives of the two problems are not easily comparable, we
suggest that our approach can provide an efficient solution of the
latter problem on graphs of small vertex cover.

\bibliography{phcross}

\end{document}